\documentclass{llncs}

\usepackage[algoruled,linesnumbered,algo2e,vlined]{algorithm2e}
\usepackage{amsmath,amssymb} 
\usepackage{graphicx}
\usepackage{wrapfig}
\usepackage[usenames]{color}
\newcommand{\ignore}[1]{}

\newcommand{\Occ}{\mathit{Occ}}
\newcommand{\NOcc}{\mathit{nOcc}}

\newcommand{\SA}{\mathit{SA}}

\newcommand{\LCP}{\mathit{LCP}}
\newcommand{\VarOcc}{\mathit{vOcc}}
\newcommand{\suffix}{\mathit{suf}}
\newcommand{\prefix}{\mathit{pre}}

\newcommand{\LnOcc}{\mathit{LnOcc}} 
\newcommand{\RnOcc}{\mathit{RnOcc}}

\newcommand{\derive}{\mathit{val}}

\newcommand{\rev}[1]{#1^R}
\newcommand{\loc}{\protect\overleftrightarrow{\mathit{loc}}}

\newcommand{\lpc}{\protect\overrightarrow{\mathit{loc}}}
\newcommand{\lsc}{\protect\overleftarrow{\mathit{loc}}}
\newcommand{\bb}{\mathit{bb}} 
\newcommand{\be}{\mathit{be}} 
\newcommand{\eb}{\mathit{eb}} 
\newcommand{\ee}{\mathit{ee}} 
\newcommand{\bstr}{\mathit{p}} 
\newcommand{\estr}{\mathit{s}}

\author{
  Keisuke Goto
  \and
  Hideo Bannai
  \and
  Shunsuke Inenaga
  \and
  Masayuki Takeda
}
\institute{
  Department of Informatics, Kyushu University \\
  744 Motooka, Nishiku, Fukuoka 819--0395, Japan \\
  \email{\{keisuke.gotou,bannai,inenaga,takeda\}@inf.kyushu-u.ac.jp}
}

\title{Computing $q$-gram Non-overlapping Frequencies on SLP Compressed Texts}

\begin{document}

\maketitle
\begin{abstract}
  Length-$q$ substrings, or $q$-grams, can represent important
  characteristics of text data,
  and determining the frequencies of all $q$-grams
  contained in the data is an important problem with many applications
  in the field of data mining and machine learning.  
  In this paper, we consider the problem of calculating the
  {\em non-overlapping frequencies} of all $q$-grams in a text
  given in compressed form, namely,
  as a straight line program (SLP).
  We show that the problem can be solved in $O(q^2n)$ time and $O(qn)$ space
  where $n$ is the size of the SLP.
  This generalizes and greatly improves previous work
  (Inenaga \& Bannai, 2009)
  which solved the problem only for $q=2$
  in $O(n^4\log n)$ time and $O(n^3)$ space.
\end{abstract}
\section{Introduction}
In many situations, large-scale text data is first compressed for storage,
and then is usually decompressed when it is processed afterwards,
where we must again face the size of the data.
To circumvent this problem, algorithms that work directly on the
compressed representation without explicit decompression have gained
attention,
especially for the string pattern matching problem~\cite{amir92:_effic_two_dimen_compr_match},
and there has been growing interest in what problems can be
efficiently solved in this kind of 
setting~\cite{lifshits07:_proces_compr_texts,navarro07:_compr,hermelin09:_unified_algor_accel_edit_distan,matsubara_tcs2009,inenaga09:_findin_charac_subst_compr_texts,goto10:_fast_minin_slp_compr_strin,philip11:_random_acces_gramm_compr_strin}.

The {\em non-overlapping occurrence frequency} of a string $P$ in
a text string $T$ is defined as the maximum number of non-overlapping
occurrences of $P$ in $T$~\cite{apostolico96:_data_struc_algo_str_sta_}.
Non-overlapping frequencies are required in several grammar based
compression
algorithms~\cite{larsson00:_off_line_diction_based_compr,apostolico00:_off_line_compr_greed_textual_subst},
as well as ...
In this paper, we consider the problem of computing the
non-overlapping occurrence frequencies of {\em all} $q$-grams
(length-$q$ substrings)
occurring in a text $T$, when the text is
given as a {\em straight line program} (SLP)~\cite{NJC97} of size $n$.
An SLP is a context free
grammar in the Chomsky normal form that derives a single string.
SLPs are a widely accepted abstract model of various text compression schemes,
since texts compressed by any grammar-based compression algorithm
(e.g.~\cite{SEQUITUR,larsson00:_off_line_diction_based_compr}) can be represented as SLPs,
and those compressed by the LZ-family (e.g.~\cite{LZ77,LZ78}) can be quickly
transformed to SLPs.
Theoretically, the length $N$ of the text represented by an
SLP of size $n$ can be as large as $O(2^n)$, and therefore 
a polynomial time algorithm that runs on an SLP representation is,
in the worst case, faster than any algorithm which works on
the uncompressed string.

For SLP compressed texts, the problem was first considered
in~\cite{inenaga09:_findin_charac_subst_compr_texts},
where an algorithm for $q=2$ running in $O(n^4\log n)$ time and
$O(n^3)$ space was presented.
However, the algorithm cannot be readily extended to handle $q>2$.
Intuitively, the problem for $q=2$ is much easier compared to 
larger values of $q$, since there is only one way for
a $2$-gram to overlap, while there can be many ways that a
longer $q$-gram can overlap.
In this paper we present the first algorithm for calculating
the non-overlapping occurrence frequency of all $q$-grams,
that works for any $q\geq2$,
and runs in $O(q^2n)$ time and $O(qn)$ space.
Not only do we solve a more general problem, but the complexity is
greatly improved compared to previous work.

A similar problem for SLPs, where occurrences of $q$-grams are allowed to overlap, was also
considered in~\cite{inenaga09:_findin_charac_subst_compr_texts},
where an $O(|\Sigma|^2n^2)$ time and $O(n^2)$ space algorithm was
presented for $q=2$. A much simpler and efficient $O(qn)$
time and space algorithm for general $q \geq 2$ was recently
developed~\cite{goto10:_fast_minin_slp_compr_strin}.
As is the case with uncompressed strings,
ideas from the algorithms allowing overlapping occurrences
can be applied {\em somewhat} to the problem of obtaining non-overlapping
occurrence frequencies. 
However, there are still
difficulties that arise from the overlapping of occurrences that must be
overcome, i.e., the occurrences of each $q$-gram can be obtained in the same
way, but we must somehow compute their non-overlapping occurrence
frequency, which is not a trivial task.

For uncompressed texts, the problem considered in this paper 
can be solved in $O(|T|)$ time,
by applying string indices such as suffix arrays.
A similar problem is the
\emph{string statistics problem}~\cite{apostolico96:_data_struc_algo_str_sta_},
which asks for the non-overlapping occurrence frequency of a given string
$P$ in text string $T$. The problem can be solved in $O(|P|)$ time for
any $P$, provided that the
text is pre-processed in $O(|T|\log |T| )$ time using the sophisticated algorithm
of~\cite{brodal02:_solvin_strin_statis_probl_time_o}. 
However, note that the preprocessing requires only $O(|T|)$ time
if occurrences are allowed to overlap.
This perhaps indicates the intrinsic difficulty that arises when
considering overlaps.

\section{Preliminaries}

\subsection{Notation}
Let $\Sigma$ be a finite {\em alphabet}.
An element of $\Sigma^*$ is called a {\em string}.
The length of a string $T$ is denoted by $|T|$. 
The empty string $\varepsilon$ is a string of length 0,
namely, $|\varepsilon| = 0$.
A string of length $q > 0$ is called a \emph{$q$-gram}.
The set of $q$-grams is denoted by $\Sigma^q$.
For a string $T = XYZ$, $X$, $Y$ and $Z$ are called
a \emph{prefix}, \emph{substring}, and \emph{suffix} of $T$, respectively.
The $i$-th character of a string $T$ is denoted by $T[i]$ for $1 \leq i \leq |T|$,
and the substring of a string $T$ that begins at position $i$ and
ends at position $j$ is denoted by $T[i:j]$ for $1 \leq i \leq j \leq |T|$.
For convenience, let $T[i:j] = \varepsilon$ if $j < i$.
Let $\rev{T}$ denote the reversal of $T$, namely, $\rev{T} = T[N]T[N-1] \cdots T[1]$,
where $N = |T|$.

For an integer $i$ and a set of integers $A$, let 
$i\oplus A = \{ i+x \mid x\in A\}$ and
$i\ominus A = \{ i-x \mid x\in A\}$.
If $A = \emptyset$, then let $i \oplus A = i\ominus A = \emptyset$.
Similarly, for a pair of integers  $(x,y)$, let $i\oplus(x,y) = (i+x,i+y)$.

\subsection{Occurrences and Frequencies}
For any strings $T$ and $P$,
let $\Occ(T,P)$ be the set of occurrences of $P$ in $T$, i.e.,
\[
 \Occ(T,P) = \{k > 0 \mid T[k:k+|P|-1] = P\}.
\]
The number of occurrences of $P$ in $T$,
or the \emph{frequency} of $P$ in $T$ is, $|\Occ(T,P)|$.
Any two occurrences $k_1,k_2 \in \Occ(T,P)$ with $k_1 < k_2$
are said to be \emph{overlapping} if $k_1 + |P| - 1 \geq k_2$.
Otherwise, they are said to be \emph{non-overlapping}.
The \emph{non-overlapping frequency} $\NOcc(T, P)$
of $P$ in $T$ is defined
as the size of a largest subset of $\Occ(T,P)$ where any
two occurrences in the set are non-overlapping.
For any strings $X,Y$, we say that an occurrence $i$ of a string $Z$ in $XY$,
with $|Z|\geq2$,
{\em crosses} $X$ and $Y$, if $i \in [|X|-|Z|+2:|X|]\cap\Occ(XY,Z)$.

For any strings $T$ and $P$, we define
the sets of \emph{right and left priority non-overlapping occurrences} of $P$ in $T$, 
respectively, as follows:
\begin{eqnarray*}
  \RnOcc(T, P) &=& \left\{
  \begin{array}{ll}
    \emptyset & \mbox{ if } \Occ(T,P) = \emptyset,\\
    \{i\} \cup \RnOcc(T[1:i-1], P) & \mbox{ otherwise, }
  \end{array}\right.\\
  \LnOcc(T, P) &=& \left\{
  \begin{array}{ll}
    \emptyset & \mbox{ if } \Occ(T,P) = \emptyset,\\
     \{ j \} \cup j \! + \! |P| \! - \! 1 \! \oplus\! \LnOcc(T[j+|P|:|T|],P)  & \mbox{ otherwise,}
  \end{array}\right.
\end{eqnarray*}
where $i = \max \Occ(T, P)$ and $j = \min \Occ(T,P)$.
For all $k \in \RnOcc(T, P)$,
it is trivially said that $\RnOcc(T[k:|T|], P) \subseteq \RnOcc(T, P)$.
It can be said to $\LnOcc$ similarly.
Note that
$\RnOcc(T,P)\subseteq \Occ(T,P)$,
$\LnOcc(T,P)\subseteq \Occ(T,P)$,
and
$\LnOcc(T, P) = |T| - |P| + 2 \ominus \RnOcc(\rev{T}, \rev{P})$.

\begin{lemma} \label{lem:left_right_max_nonoverlap}
  $\NOcc(T, P) = |\RnOcc(T, P)| = |\LnOcc(T, P)|$
\end{lemma}

\begin{proof}
  See Appendix.
\end{proof}

\begin{lemma}
  \label{lemma:nocc}
  For any strings $T$ and $P$, and any integer $i$ with $1 \leq i \leq |T|$,
  let $\mathit{u_1} = \max\LnOcc(T[1:i-1],P)+|P|-1$ and
  $\mathit{u_2} = i-1 + \min\RnOcc(T[i:|T|],P)$.
  Then $\NOcc(T,P) = |\LnOcc(T[1:u_1], P)| + \NOcc(T[u_1+1:u_2-1], P) + |\RnOcc(T[u_2:|T|], P)|$.
\end{lemma}

\begin{proof}
By Lemma~\ref{lem:left_right_max_nonoverlap} and the definitions of
$u_1$, $u_2$, $\LnOcc$ and $\RnOcc$, we have 
\begin{eqnarray*}
 \lefteqn{\NOcc(T,P)} \\ 
 & = & |\LnOcc(T[1:u_1], P)| + |\LnOcc(T[u_1+1:|T|], P)| \\
 & = & |\LnOcc(T[1:u_1], P)| + |\RnOcc(T[u_1+1:|T|], P)| \\
 & = & |\LnOcc(T[1:u_1], P)| \! + \! |\RnOcc(T[u_1\!+\!1:u_2\!-\!1], P)| \! + \! |\RnOcc(T[u_2:|T|], P)| \\
 & = & |\LnOcc(T[1:u_1], P)| + \NOcc(T[u_1\!+\!1:u_2-1], P) + |\RnOcc(T[u_2:|T|], P)|.
\end{eqnarray*}
\qed
\end{proof}

We will later make use of the solution to the following problem,
where occurrences of $q$-grams are weighted and allowed to overlap.

\begin{problem}[weighted overlapping $q$-gram frequencies] \label{prob:weighted-q-gram_frequencies}
  Given a string $T$, an integer $q$, and integer array $w$ ($|w| = |T|$),
  compute $\sum_{i \in Occ(T, P)} w[i]$ for all $q$-grams $P \in
  \Sigma^q$ where $\Occ(T,P)\neq\emptyset$.
\end{problem}

\begin{theorem}[\cite{goto10:_fast_minin_slp_compr_strin}]\label{thm:weighted_q_gram}
  Problem~\ref{prob:weighted-q-gram_frequencies} can be solved in $O(|T|)$ time.
\end{theorem}
\begin{proof}
  See Appendix.
\end{proof}

\subsection{Straight Line Programs}
In this paper, we treat strings described in terms of 
\emph{straight line programs} ({\em SLPs}).
A straight line program $\mathcal T$ is a sequence of assignments
$\{ X_1 = expr_1,$ $X_2 = expr_2, \ldots, X_n = expr_n\}$.
Each $X_i$ is a variable and each $expr_i$ is an 
expression where $expr_i = a$ ($a\in\Sigma$), or
$expr_i = X_{\ell} X_r$ ($\ell,r < i $).
We will sometimes abuse notation and denote $\mathcal{T}$ as $\{X_i\}_{i=1}^n$.
Denote by $T$ the string derived from the last variable $X_n$
of the program $\mathcal T$.
Fig.~\ref{fig:SLP} shows an example of an SLP.
The \emph{size} of the program $\mathcal T$ is the number $n$ of
assignments in $\mathcal T$.

\begin{wrapfigure}[15]{r}{0.5\textwidth}
   \vspace{-1.0cm}
  \centerline{\includegraphics[width=0.5\textwidth]{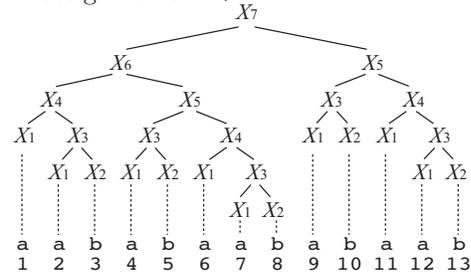}}
  \caption{
    The derivation tree of
    SLP $\mathcal{T} =
    \{
    X_1 = \mathtt{a},
    X_2 = \mathtt{b},
    X_3 = X_1X_2, 
    X_4 = X_1X_3,
    X_5 = X_3X_4,
    X_6 = X_4X_5,
    X_7 = X_6X_5\}$,
    which represents string $T = \derive(X_7) = \mathtt{aababaababaab}$.
  }
  \label{fig:SLP}
\end{wrapfigure}

Let $\derive(X_i)$ represent the string derived from $X_i$.
When it is not confusing, we identify a variable $X_i$
with $\derive(X_i)$.
Then, $|X_i|$ denotes the length of the string $X_i$ derives,
and
$X_i[j] = \derive(X_i)[j]$, 
$X_i[j:k] = \derive(X_i)[j:k]$ for $1 \leq j,k\leq |X_i|$.
Let $\VarOcc(X_i)$ denote the number of times a variable $X_i$ occurs in the derivation of $T$.
For example,  $\VarOcc(X_4) = 3$ in Fig.~\ref{fig:SLP}.

Both $|X_i|$ and $\VarOcc(X_i)$ can be computed 
for all $1 \leq i \leq n$ in a total of $O(n)$ time by a simple
iteration on the variables:
$|X_i| = 1$ for any $X_i = a~ (a \in \Sigma)$, and
$|X_i| = |X_\ell|+|X_r|$ for any $X_i = X_\ell X_r$.
Also, $\VarOcc(X_n) = 1$ and for $i < n$,
$\VarOcc(X_i) =
\sum \{ \VarOcc(X_k) \mid X_k = X_\ell X_i \}
+
\sum \{\VarOcc(X_k) \mid X_k = X_iX_r \}$.

  We shall assume as in various previous work on SLP, 
  that the word size is at least $\log |T|$, and hence,
  values representing lengths and positions of $T$
  in our algorithms can be manipulated in
  constant time.

\section{$q$-gram Non-Overlapping Frequencies on Compressed String}

The goal of this paper is to efficiently solve the following problem.
\begin{problem}[Non-overlapping $q$-gram frequencies on SLP] \label{prob:q-gram_frequencies}
Given an SLP $\mathcal{T}$ of size $n$ 
that describes string $T$ and a positive integer $q$, 
compute $\NOcc(T, P)$ for all $q$-grams $P \in \Sigma^q$.
\end{problem}

If we decompress the given SLP $\mathcal{T}$ obtaining the string $T$,
then we can solve the problem in $O(|T|)$ time.
However, 
it holds that $|T| = O(2^n)$.
Hence, in order to solve the problem efficiently,
we have to establish an algorithm that does not explicitly decompress 
the given SLP $\mathcal{T}$.

\subsection{Key Ideas}

For any variable $X_i$ and integer $k \geq 1$, 
let $\prefix(X_i,k) = X_i[1:\min\{k, |X_i|\}]$ and 
$\suffix(X_i, k) = X_i[|X_i| - \min\{k, |X_i|\} + 1:|X_i|]$.
That is, $\prefix(X_i,k)$ and $\suffix(X_i, k)$ are 
the prefix and the suffix of $\derive(X_i)$ of length $k$, respectively.
For all variables $X_i$,
$\prefix(X_i,k)$ can be computed in a total of $O(nk)$ time and space, as follows:
\[
 \prefix(X_i, k) = 
  \begin{cases}
   \derive(X_i) & \mbox{if } |X_i| \leq k, \\
   \prefix(X_\ell, k)\prefix(X_r, k-|X_\ell|) & \mbox{if } X_i = X_\ell X_r \mbox{ and } |X_\ell| < k < |X_i|, \\
   \prefix(X_\ell, k) & \mbox{if } X_i = X_\ell X_r \mbox{ and } k \leq |X_\ell|.
  \end{cases}
\]
$\suffix(X_i, k)$ can be computed similarly in $O(nk)$ time and space.

For any string $T$ and positive integers $q$ and $j$ ($1 \leq j \leq j+q-1 \leq |T|$),
the \emph{longest overlapping cover} of the $q$-gram $P=T[j:j+q-1]$ w.r.t. position $j$ of $T$ is 
an ordered pair $\loc_{q}(T, j) = (b,e)$ of positions in $T$
which is defined as:
\begin{eqnarray*}
&&\loc_{q}(T,j) = \arg\max_{(b,e)}\\
&&
\left\{
 (e-b) ~\left\vert~
\begin{array}{l}
  (b,e)\in\Occ(T,P)\times((q-1)\oplus\Occ(T,P)),\\
  {b\leq j\leq j+q-1\leq e},\\
  \forall k \in [b:e-q]\cap\Occ(T,P),\\
  ~~~~     [k+1:\min\{k+q-1,e-q+1\}]\cap\Occ(T,P)\neq\emptyset
\end{array} \right\}\right.
\end{eqnarray*}
Namely, 
$\loc_q(T, j)$ represents 
the beginning and ending positions of the maximum chain of overlapping occurrences of $q$-gram $T[j:j+q-1]$
that contains position $j$.
For example, consider string $T = \mathtt{aaabaabaaabaabaaaabaa}$ of length 21.
For $q = 5$ and $j = 9$,
we have $\loc_q(T, j) = (2, 16)$, since $T[2:6] = T[5:9] = T[9:13] = T[12:16] = \mathtt{aabaa}$.
Note that $T[17:21] = \mathtt{aabaa}$ is not contained in this chain since 
it does not overlap with $T[12:16]$.

\begin{lemma}
  \label{lem:longest_cover}
Given a string $T$ and integers $q, j$,
the longest overlapping cover $\loc_{q}(T, j)$ can be computed in $O(|T|)$ time.
\end{lemma}
\begin{proof}
Using, for example, the KMP algorithm~\cite{KMP77},
we can obtain a sorted list of $\Occ(T,T[j:j+q-1])$ in $O(|T|)$ time.
We can just scan this list forwards and backwards, to easily obtain $b$ and $e$.
\qed
\end{proof}

For a variable $X_i = X_\ell X_r$ and a position $1 \leq j \leq |X_i|-q+1$,
a longest overlapping cover $(b, e) = \loc_q(X_i, j)$ is said to be \emph{closed in $X_i$}
if
$q - 1 < b$  and $e < |X_i| - q + 2$.

\begin{theorem} \label{theo:q-gram}
  Problem~\ref{prob:q-gram_frequencies} can be solved in $O(q^2n)$ time,
  provided that, for all variables $X_i=X_\ell X_r$ and $j$ s.t.
  $|X_i| \geq q$ and
  $\max\{1,|X_\ell|-2(q-1)+1\} \leq j \leq \min\{|X_\ell|+q-1,|X_i|-q+1\}$,
  $(b, e) = \loc_q(X_i, j)$ and $\NOcc(X_i[b:e],s)$
  are already computed where $s = X_i[j:j+q-1]$.
\end{theorem}

\begin{proof}
  Algorithm~\ref{algo:slpmain} shows a pseudo-code of our algorithm 
  to solve Problem~\ref{prob:q-gram_frequencies}.

  Consider $q$-gram $s = X_i[j:j+q-1]$ at position $j$
  for which $(b,e) = \loc_q(X_i, j)$ is closed in $X_i$.
  A key observation is that, if $(b,e)$ is closed in $X_i$, 
  then $(b,e)$ is never closed in $X_\ell$ or $X_r$.
  Therefore, by summing up $\VarOcc(X_i) \cdot \NOcc(X_i[b:e],s)$
  for each closed $(b,e)$ in $X_i$, for all such variables $X_i$, we
  obtain $\NOcc(T, s)$.
  Line~\ref{line:closed} 
  is sufficient to check if $(b,e)$ is closed.
  
  For all $1 \leq i \leq n$, 
  $\VarOcc(X_i)$ can be computed in $O(n)$ time,
  and 
  $t_i = \prefix(X_i, 2(q-1))\suffix(X_i, 2(q-1))$ 
  can be computed in $O(qn)$ time and space.
  The problem amounts to summing up the values of
  $\VarOcc(X_i) \cdot \NOcc(X_i[b:e],s)$ for each $q$-gram $s$
  contained in each $t_i$,
  and can be reduced to Problem~\ref{prob:weighted-q-gram_frequencies}
  on string $z$ and integer array $w$ of length $O(qn)$,
  which can be solved in $O(qn)$ time by Theorem~\ref{thm:weighted_q_gram}.

  In line~\ref{line:condition_avoid_double},
  we check if there is no previous position $h$ ($\max\{1, |X_\ell|-2(q-1)+1\} \leq h < j$) such that
  $X_i[h:h+q-1] = X_i[j:j+q-1]$
  by $\loc_q(X_i, h) = \loc_q(X_i, j)$,
  so that we do not count the same $q$-gram more than once.
  If there is no such $h$, we set the value of $w_i[k-|X_\ell|+j]$ to 
  $\VarOcc(X_i) \cdot \NOcc(X_i[b:e],s)$.
  This can be checked in $O(q^2n)$ time for all $X_i$ and $j$.
  
  For convenience, we assume that $T=\derive(X_n)$ starts and ends
  with special characters
  $\#^{q-1}$ and $\$^{q-1}$ that do not occur anywhere else in $T$, respectively.
  Then we can cope with the last variable $X_n$ as described above.
  Hence the theorem holds.
  \qed

\end{proof}

\begin{algorithm2e}[t]
  \caption{Computing $q$-gram non-overlapping frequencies from SLP}
  \label{algo:slpmain}
  \SetKw{KwAnd}{and}
  \SetKw{KwOr}{or}
  \SetKw{KwTrue}{true}
  \SetKw{KwFalse}{false}
  \KwIn{SLP ${\mathcal T} = \{X_i\}_{i=1}^n$ representing string $T$, integer $q\geq 2$.}
  \KwOut{$\NOcc(T, P)$ for all $q$-grams $P \in \Sigma^q$ where $\Occ(T,P)\neq\emptyset$.}
  Compute $\VarOcc(X_i)$ for all $1\leq i\leq n$\; 
  Compute $\prefix(X_i,2(q-1))$ and $\suffix(X_i,2(q-1))$ for all $1\leq
  i\leq n-1$\; 
  $z \leftarrow \varepsilon$; $w \leftarrow []$\;
  \For{$i \leftarrow 1$ \KwTo $n$}{
   \If{$|X_i| \geq q$}{
    let $X_i = X_\ell X_r$\;
    $k \leftarrow |\suffix(X_\ell,2(q-1))|$\;
    $t_i = \suffix(X_\ell,2(q-1))\prefix(X_r,2(q-1))$\;
    $z$.append($t_i$)\;
    $w_i \leftarrow$ create integer array of length $|t_i|$, each element set to $0$\;
    \For{$j \leftarrow \max\{1, |X_\ell|-2(q-1)+1\}$ \KwTo $\min\{|X_\ell|+q-1, |X_i|-q+1\}$\label{line:loop_xi}}{
      $s \leftarrow X_i[j:j+q-1]$\;
      $(b, e) \leftarrow \loc_q(X_i, j)$\; \label{line:be_xi}
      \If{$q-1 < b \mbox{ \KwAnd\ } e < |X_i|-q+2$\label{line:closed}}{
        
        \If{$\loc_q(X_i, h) \neq \loc_q(X_i, j)$ for any position $h$ s.t. $\max\{1, |X_\ell|-2(q-1)+1\} \leq h < j$\label{line:condition_avoid_double}}{
        $w_i[k - |X_\ell| + j] \leftarrow \VarOcc(X_i) \cdot \NOcc(X_i[b:e], s)$\; \label{line:add}
       }
      }
     }
    }
    $w$.append($w_i$)\;
   }
  Calculate $q$-gram frequencies in $z$,
  where each $q$-gram starting at position $d$ is
  {\em weighted} by $w[d]$.\label{algo:slpqgrammainweightedfreqs} 
\end{algorithm2e}

\subsection{Computing Longest Overlapping Covers}

In this subsection, we will show how 
to compute longest overlapping cover $(b, e) = \loc_q(X_i, j)$ 
where $s = X_i[j:j+q-1]$ for all $X_i$ and
all $j$ required for Theorem~\ref{theo:q-gram}.

For any string $T$ and integers $q$ and $j$ ($1 \leq j < q$), let
\begin{eqnarray*}
 \lpc_q(T, j) & = & 
  \begin{cases}
   (j, \be) & \mbox{if } j + q-1 \leq |T|, \\
   (j,|T|) & \mbox{otherwise},
  \end{cases} \\
  \lsc_q(T, j) & = &
  \begin{cases}
   (\eb, |T|-j+1) & \mbox{if } |T|-j-q+2 \geq 1, \\
   (1, |T|-j+1) & \mbox{otherwise},
  \end{cases}   
\end{eqnarray*}
where $(j, \be) = (j-1) \oplus \loc_q(T[j:|T|], 1)$
and $(\eb, |T|-j+1) = \loc_q(T[1:|T|-j+1], |T|-j-q+2)$.
Namely, $\lpc_q(T, j)$ is a suffix of the longest overlapping cover 
of the $q$-gram $T[j:j+q-1]$ that begins at position $j$ ($1 \leq j < q$) in $T$,
and $\lsc_q(T, j)$ is a prefix of the longest overlapping cover 
of the $q$-gram $T[|T|-j-q+2:|T|-j+1]$ that ends at position $|T|-j+1$ in $T$.

\begin{lemma} \label{lem:lpc}
For all $1 \leq i \leq n$ and $1 \leq j \leq 2(q-1)$,
$\lpc_q(X_i, j)$ can be computed in a total of $O(q^2n)$ time.
\end{lemma}
\begin{proof}

  We use dynamic programming.
  Let $X_i = X_\ell X_r$, $\bstr_j=X_i[j:j+q-1]$, and assume 
  $\lpc_q(X_\ell,j)$ and $\lpc_q(X_r,j)$ have been calculated for
  all $1 \leq j \leq 2(q-1)$.
  We examine the string $X_i[\max\{j,|X_\ell|-q+2\}:\min\{|X_i|,|X_\ell|+q-1\}]$ for
  occurrences of $\bstr_j$ that cross $X_\ell$ and $X_r$,
  obtain its longest overlapping cover
  $(b_i,e_i)$, and check if it overlaps with $\lpc_q(X_\ell,j)$.
  Furthermore, let $\bb_r$ be the left most occurrence of 
  $\bstr_j$ in $X_r$ that has the possibility of overlapping with
  $(b_i,e_i)$.
  Then, $\lpc_q(X_i,j)$ is either $\lpc_q(X_\ell,j)$,
  or its end can be extended to $e_i$, or further to the end of
  $\lpc_q(X_r,\bb_r)$,
  depending on how the covers overlap.

  More precisely, let 
  $(j, \be_\ell) = \lpc_q(X_\ell, j)$,
  $(b_i, e_i) = \max\{j-1, |X_\ell|-q+1\} \oplus \loc_q(X_i[\max\{j, |X_\ell|-q+2\}:\min\{|X_i|, |X_\ell|+q-1\}],h)$
  where $h\in\Occ(X_i[\max\{j,$ $|X_\ell|-q+2\}:\min\{|X_i|, |X_\ell|+q-1\}],\bstr_j)$,
  and
  $(\bb_r, \be_r) = (|X_\ell|+k-1) \oplus\lpc_q(X_r, k)$ where $k = \min \Occ(\prefix(X_r, 2(q-1)), \bstr_j)$.
  (Note that $(\bb_r, \be_r), (b_i, e_i)$ are not
  defined if occurrences $h, k$ of $\bstr_j$ do not exist.)
  Then we have
  \[
  \lpc_q(X_i, j) = 
  \begin{cases}
    (j, \be_\ell) & \mbox{if } \be_\ell < b_i \mbox{ or } \not\exists h,\\
    (j, e_i) & \mbox{if } b_i \leq \be_\ell \mbox{ and } (e_i < \bb_r \mbox{ or } \not\exists k)\\
    (j, \be_r) & \mbox{otherwise}.
  \end{cases}
  \]
  (See also Fig.~\ref{fig:ploc} in Appendix.)
  For all variables $X_i$ we pre-compute $\prefix(X_i, 2(q-1))$ and $\suffix(X_i, 2(q-1))$.
  This can be done in a total of $O(qn)$ time. 
  Then, each $\lpc_q(X_i, j)$ can be computed in $O(q)$ time
  using the KMP algorithm, Lemma~\ref{lem:longest_cover},
  and the above recursion, giving a total of $O(q^2n)$
  time for all $1 \leq i \leq n$ and $1 \leq j \leq 2(q-1)$.
  \qed
\end{proof}
\begin{lemma} \label{lem:lsc}
For all $1 \leq i \leq n$ and $1 \leq j \leq 2(q-1)$,
$\lsc_q(X_i, j)$ can be computed in a total of $O(q^2n)$ time.
\end{lemma}
\begin{proof}
  The proof is essentially the same as the proof for $\lpc_q(X_i,j)$ in Lemma~\ref{lem:lpc}.  
\end{proof}

Recall that we have assumed in Theorem~\ref{theo:q-gram}
that $\loc_q(X_i, j)$ are already computed.
The following lemma describes how $\loc_q(X_i, j)$ can actually be computed 
in a total of $O(q^2n)$ time.
\begin{lemma} \label{lem:loc_kusi}
  For all variable $X_i=X_\ell X_r$ and $j$ s.t.
  $\max\{1, |X_\ell|-2(q-1)+1\} \leq j \leq \min\{|X_\ell|+q-1, |X_i|-q+1\}$,
  $(b, e) = \loc_q(X_i, j)$  
  can be computed in a total of $O(q^2n)$ time.
\end{lemma}

\begin{proof}
  Let $s_j = X_i[j:j+q-1]$.
  Firstly, we compute $(b_i, e_i) =
  \loc_q(X_i[|X_\ell|-2(q-1)+1:\min\{|X_i|,|X_\ell|+2(q-1)\}], j)$
  and then $\loc_q(X_i,j)$ can be computed based on $(b_i,e_i)$, as follows:
  Let $(\eb_\ell, \ee_\ell) = \lsc_q(X_\ell, |X_\ell|-\ee_\ell+1)$ and
  $(\bb_r, \be_r) = |X_\ell| \oplus \lpc_q(X_r, \bb_r-|X_\ell|)$,
  where
  $\ee_\ell = \max Occ(X_i[\max\{1, |X_\ell|-2(q-1)+1\}:|X_\ell|], s_j)$
  and $\bb_r = \min Occ(X_i[|X_\ell|+1:\min\{|X_i|,|X_\ell|+2(q-1)\}], s_j)$.

  \begin{enumerate}
  \item 
  If $b_i \leq |X_\ell|$ and $e_i > |X_\ell|$, 
  then we have $b \leq b_i \leq |X_\ell| < e_i \leq e$.
  $(b, e) = \loc_q(X_i, j)$ can be computed by checking  
  whether $(\eb_\ell, \ee_\ell)$, $(b_i, e_i)$, and $(\bb_r, \be_r)$ are overlapping or not.
  (See also Fig.~\ref{fig:loc_kusi} in Appendix.)

  \item 
  If $e_i \leq |X_\ell|$, then trivially $b = \eb_\ell$ and $e = e_i$.

  \item   
  If $b_i > |X_\ell|$, then trivially $b = b_i$ and $e = \be_r$.
  \end{enumerate}
  Each $\ee_\ell = h$ and $\bb_r = |X_\ell| + k$ can be computed 
  using the KMP algorithm on string $\suffix(X_\ell, 2(q-1))\prefix(X_r, 2(q-1))$ in $O(q)$ time.
  By Lemmas~\ref{lem:lpc} and~\ref{lem:lsc}, 
  $(\eb_\ell, \ee_\ell)$ and $(\bb_r, \be_r)$ can be pre-computed 
  in a total of $O(q^2n)$ time for all $1 \leq i \leq n$.
  Hence the lemma holds.
  \qed
\end{proof}

\subsection{Largest Left-Priority and Smallest Right-Priority Occurrences}

In order to compute $\NOcc(X_i[b:e], s)$
for all $X_i$ and all $j$ required for Theorem~\ref{theo:q-gram},
where $(b, e) = \loc_q(X_i, j)$ and $s = X_i[j:j+q-1]$,
we will use the largest and second largest occurrences of $\LnOcc$ and 
the smallest and second smallest occurrences of $\RnOcc$.

For any set $S$ of integers and integer $1 \leq k \leq |S|$, let 
$\max_k S$ and $min_k S$ denote the $k$-th largest and the $k$-th smallest
element of $S$.

For $1 \leq i \leq n$ and $1 \leq j \leq 2(q-1)$,
consider to compute $\max_k \LnOcc(X_i[j:\be_i], \bstr_{j})$ for $k = 1,2$,
where $(j, \be_i) = \lpc_q(X_i, j)$ and $\bstr_{j} = X_i[j:j+q-1]$.
Intuitively, difficulties in computing $\max_k \LnOcc(X_i[j:\be_i], \bstr_{j})$
come from the fact that the string $\derive(X_i)[j:\be_i]$ can be as long as $O(2^n)$,
but we only have prefix $\prefix(X_i,3(q-1))$ and suffix $\suffix(X_i,3(q-1))$
of $\derive(X_i)$ of length $O(q)$.
Hence we cannot compute the value of $\be_i$ by simply running the KMP algorithm
on those partial strings.
For the same reason, the size of $\LnOcc(X_i[j:\be_i], \bstr_{j})$ can be as large as $O(2^n/q)$.
Hence we cannot store $\LnOcc(X_i[j:\be_i], \bstr_{j})$ as is.
Still, as will be seen in the following lemma,
we can compute those values efficiently, only in $O(q^2n)$ time.

\begin{lemma} \label{lem:max_lnocc_lpc}
 For all variable $X_i=X_\ell X_r$ and $1 \leq j \leq 2(q-1)$,
 let $(j, \be_i) = \lpc_q(X_i, j)$, $\bstr_{j} = X_i[j:j+q-1]$.
 
 We can compute the values
 $\max_1 \LnOcc(X_i[j:\be_i], \bstr_{j})$ and
 $\max_2 \LnOcc(X_i[j:\be_i], \bstr_{j})$

 for all $1 \leq i \leq n$ and $1 \leq j \leq 2(q-1)$,
 in a total of $O(q^2n)$ time.  
\end{lemma}

\begin{proof}
See Appendix.

\end{proof}

The next lemma can be shown similarly to Lemma~\ref{lem:max_lnocc_lpc}.
\begin{lemma} \label{lem:min_rnocc_eloc}
 For all variable $X_i=X_\ell X_r$ and $1 \leq j \leq 2(q-1)$,
 let $(\eb, \ee) = \lsc_q(X_i, j)$, and $\estr_{j} = X_i[|X_i|-j-q+2:|X_i|-j+1]$.
 We can compute the values
 $\min_1 \RnOcc(X_i[\eb:\ee], \estr_{j})$ and
 $\min_2 \RnOcc(X_i[\eb:\ee], \estr_{j})$
 for all $1 \leq i \leq n$ and $1 \leq j \leq 2(q-1)$,
 in a total of $O(q^2n)$ time.
\end{lemma}

\begin{lemma} \label{lem:max_lnocc_lsc}
  For all variable $X_i=X_\ell X_r$ and $1 \leq j < q$,
  $\max \LnOcc(X_i[\eb_i:\ee_i], \estr_{j})$ can be computed in a total of $O(q^2n)$ time,
  where $(\eb_i, \ee_i) = \lsc_q(X_i, j)$ and $\estr_{j} = X_i[|X_i|-j-q+2:|X_i|-j+1]$.
\end{lemma}

\begin{proof}
The lemma can be shown by using Lemma~\ref{lem:max_lnocc_lpc}.  
See Appendix for details.
\end{proof}

\begin{lemma} \label{lem:min_rnocc_lpc}
  For all variable $X_i=X_\ell X_r$ and $1 \leq j < q$,
  $\min \RnOcc(X_i[\bb_i:\be_i], \bstr_{j})$ can be computed in a total of $O(q^2n)$ time,
  where $(\bb_i, \be_i) = \lpc_q(X_i, j)$ and $\bstr_{j} = X_i[j:j+q-1]$.
\end{lemma}

\begin{proof}
The lemma can be shown in a similar way to
Lemma~\ref{lem:max_lnocc_lsc},
using Lemma~\ref{lem:min_rnocc_eloc} instead of Lemma~\ref{lem:max_lnocc_lpc}.
\qed
\end{proof}

\subsection{Counting Non-Overlapping Occurrences in Longest Overlapping Covers}

Firstly, we show how to count
non-overlapping occurrences of $q$-gram $p_j$ in $X_i[j:\be_i]$,
for all $i$ and $j$,
where $p_j = X_i[j:j+q-1]$ and $(j, \be_i) = \lpc_q(X_i,j)$.

\begin{lemma} \label{lem:nocc_lpc}
 For all variable $X_i=X_\ell X_r$ and $1 \leq j \leq 2(q-1)$,
 let $(j, \be_i) = \lpc_q(X_i, j)$ and $\bstr_{j} = X_i[j:j+q-1]$.
 We can compute $\NOcc(X_i[j:\be_i], \bstr_{j})$
 for all $1 \leq i \leq n$ and $1 \leq j \leq 2(q-1)$,
 in a total of $O(q^2n)$ time.
\end{lemma}

\begin{proof}
By Lemma~\ref{lem:left_right_max_nonoverlap},
we have $\NOcc(X_i[j:\be_i], \bstr_{j}) = |\LnOcc(X_i[j:\be_i], \bstr_{j})|$.
We compute the occurrence $b_i$ in $(j-1) \oplus \LnOcc(X_i[j:\be_i], \bstr_{j})$
that crosses $X_\ell$ and $X_r$, if such exists.
Note that at most one such occurrence exists.
Also, we compute the smallest occurrence $\bb_r$ in $(j-1) \oplus \LnOcc(X_i[j:\be_i], \bstr_{j})$
that is completely within $X_r$.
Then the desired value $\NOcc(X_i[j:\be_i], \bstr_{j})$ can be computed
depending whether $b_i$ and $\bb_r$ exist or not.

Formally:
Consider the set 
$S = ((j-1) \oplus \LnOcc(X_i[j:\be_i], \bstr_j)) \cap [|X_\ell|-q+2:|X_\ell|]$
of occurrence of $\bstr_j$ which is either empty or singleton.
If $S$ is singleton, then let $b_i$ be its single element.
Let $\bb_r = \min \{k \mid k \in ((j-1) \oplus \LnOcc(X_i[j:\be_i], \bstr_j)) \cap [|X_\ell|+1:|X_\ell|+q-1],
\mbox{if } \exists b_i \mbox{ then } k \geq b_i + q\}$.

Then we have
\begin{eqnarray*}
\lefteqn{\NOcc(X_i[j:\be_i], \bstr_j)} \\
& = & 
 \begin{cases}
  \NOcc(X_r[j-|X_\ell|:\be_i-|X_\ell|], \bstr_j) & \mbox{if } j > |X_\ell|, \\
  \NOcc(X_\ell[j:\be_\ell], \bstr_j) & \mbox{if} \not\exists b_i \mbox{ and} \not\exists \bb_r,\\
  \NOcc(X_\ell[j:\be_\ell], p_j) + 1 & \mbox{if } \exists b_i \mbox{ and} \not\exists \bb_r \\
  \NOcc(X_\ell[j:\be_\ell], p_j) + \NOcc(X_r[b_r:\be_r], p_j) & \mbox{if} \not\exists b_i \mbox{ and } \exists \bb_r, \\
  \NOcc(X_\ell[j:\be_\ell], p_j) + \NOcc(X_r[b_r:\be_r], p_j) + 1 & \mbox{if } \exists b_i \mbox{ and } \exists \bb_r, \\
 \end{cases}
\end{eqnarray*}
where $(\bb_r, \be_r) = \lpc_q(X_r, \bb_r)$.

For all variables $X_i$ we pre-compute $\prefix(X_i, 3(q-1))$ and $\suffix(X_i, 3(q-1))$.
This can be done in a total of $O(qn)$ time.
If $b_i$ or $\bb_r$ exists,
$|X_\ell|-3(q-1) < j-1+ \max \LnOcc(X_\ell[j:\be_\ell], j) \leq |X_\ell|-q+2$.
Then, each $b_i$ and $\bb_r$ can be computed from
$\LnOcc(X_i[(j-1+ \max \LnOcc(X_\ell[j:\be_\ell], j)):|X_\ell|+3(q-1)], p_j)$
running the KMP algorithm on string $\suffix(X_\ell, 3(q-1))\prefix(X_r, 3(q-1))$.
Based on the above recursion, we can compute 
$\NOcc(X_i[j:\be_i], \bstr_j)$ in a total of $O(q^2n)$
time for all $1 \leq i \leq n$ and $1 \leq j \leq 2(q-1)$.
\qed
\end{proof}

The next lemma can be shown similarly to Lemma~\ref{lem:nocc_lpc}.
\begin{lemma} \label{lem:nocc_lsc}
 For all variable $X_i=X_\ell X_r$ and $1 \leq j \leq 2(q-1)$,
 let $(\eb_i, \ee_i) = \lsc_q(X_i, j)$ and $\estr_{j} = X_i[|X_i|-j-q+2:|X_i|-j+1]$.
 We can compute  $\NOcc(X_i[\eb_i:\ee_i], \estr_{j})$
 for all $1 \leq i \leq n$ and $1 \leq j \leq 2(q-1)$,
 in a total of $O(q^2n)$ time.
\end{lemma}

We have also assumed in Theorem~\ref{theo:q-gram} that 
$\NOcc(X_i[b:e], s_j)$ are already computed.
This can be computed efficiently, as follows:
\begin{lemma} \label{lem:nocc_in_loc}
  For all variable $X_i=X_\ell X_r$ and $j$ s.t.  
  $\min\{1, |X_\ell|-2(q-1)+1\} \leq j \leq \min\{|X_i|-q+1,|X_\ell|+q-1\}$,
  $\NOcc(X_i[b:e], s_j)$  can be computed in a total of $O(q^2n)$ time,
  where $(b,e)=\loc_q(X_i, j)$ and $s_j=X_i[j:j+q-1]$.
\end{lemma}
\begin{proof}
 
We consider the case where $\max\{1, |X_\ell|-q+2\} \leq j \leq |X_\ell|$,
as the other cases can be shown similarly.
Our basic strategy for computing $\NOcc(X_i[b:e], s_j)$
is as follows.
Firstly we compute the largest element of $\LnOcc(X_i[b:e], s_j)$
that occurs completely within $X_\ell$.
Secondly we compute the smallest element of $\RnOcc(X_i[b:e], s_j)$
that occurs completely within $X_r$.
Thirdly we compute an occurrence of $s_j$ 
that crosses the boundary of $X_\ell$ and $X_r$,
and do not overlap the above occurrences of $\estr_j$ completely within $X_\ell$ and $X_r$.

Formally:
Let $\ee_\ell= b + q - 2 + \max Occ(X_i[b:|X_\ell|], s_j)$,
$\bb_r = |X_\ell|+\min Occ(X_i[|X_\ell|+1:e], s_j)$,
$u_1 = b+q-2+\max \LnOcc(X_i[b:\ee_\ell], s_j)$, and
$u_2 = \bb_r-1+\min \RnOcc(X_i[\bb_r:e], s_j)$.
We consider the case where all these values exist,
as other cases can be shown similarly.
It follows from Lemmas~\ref{lem:left_right_max_nonoverlap} and~\ref{lemma:nocc} that
\begin{eqnarray*}
  \lefteqn{\NOcc(X_i[b:e], s_j)} \\
  & = & |\LnOcc(X_i[b:u_1], s_j)| \! + \! \NOcc(X_i[u_1\!+\!1:u_2\!-\!1], s_j) \! + \! |\RnOcc(X_i[u_2:e], s_j)| \\
  & = & \NOcc(X_i[b:\ee_\ell], s_j) + \NOcc(X_i[u_1+1:u_2-1], s_j) + \NOcc(X_i[\bb_r:e], s_j), 
\end{eqnarray*}
(See also Fig.~\ref{fig:nocc_loc} in Appendix.)

By Lemma~\ref{lem:loc_kusi},
$(b, e) = \loc_q(X_i, j)$ can be pre-computed in a total of $O(q^2n)$ time.
Since $b < \ee_\ell$ and $\bb_r < e$,
$\ee_\ell$ and $\bb_r$ can be computed in $O(q)$ time 
using the KMP algorithm.
By Lemmas~\ref{lem:nocc_lpc} and~\ref{lem:nocc_lsc}
$\NOcc(X_i[b:\ee_\ell], s_j)$ and $\NOcc(X_i[\bb_r:e], s_j)$
can be pre-computed in a total of $O(q^2n)$ time
(Notice $(b, \ee_\ell) = \lsc_q(X_\ell, \ee_\ell)$
and $(\bb_r,e) = |X_\ell| \oplus \lpc_q(X_r, \bb_r-|X_\ell|)$).
By Lemmas~\ref{lem:max_lnocc_lsc} and~\ref{lem:min_rnocc_lpc},
$u_1$ and $u_2$ can be pre-computed in a total of $O(q^2n)$ time.
Hence $\NOcc(X_i[u_1+1:u_2-1], s_j)$ can be computed 
in $O(q)$ time using the KMP algorithm
for each $i$ and $j$.
The lemma thus holds.
\qed

\end{proof}

\subsection{Main Result}

The following theorem concludes this whole section.
\begin{theorem}
Problem~\ref{prob:q-gram_frequencies} can be solved in $O(q^2n)$ time and $O(qn)$ space.
\end{theorem}

\begin{proof}
The time complexity and correctness follow from 
Theorem~\ref{theo:q-gram}, Lemma~\ref{lem:loc_kusi}, and Lemma~\ref{lem:nocc_in_loc}.

We compute and store 
strings $\suffix(X_i, 3(q-1))$ and $\prefix(X_i, 3(q-1))$ of length $O(q)$
for each variable $X_i$, hence this requires a total of $O(qn)$ space
for all $1 \leq i \leq n$.
We use a constant number of dynamic programming tables
each of which is of size $O(qn)$.
Hence the total space complexity is $O(qn)$.
\qed
\end{proof}

\section{Conclusion and Discussion}
We considered the problem of computing the non-overlapping frequencies
for all $q$-grams that occur in a given text represented as an SLP.
Our algorithm greatly improves previous work
which solved the problem only for $q=2$ requiring $O(n^4\log n)$ time and $O(n^3)$ space.
We give the first algorithm which works for any $q\geq 2$,
running in $O(q^2n)$ time
and $O(qn)$ space, where $n$ is the size of the SLP.

\bibliographystyle{splncs03}
\bibliography{ref}

\clearpage
\appendix
\section*{Appendix}

\section{Proofs}

\subsubsection{Proof of Theorem~\ref{thm:weighted_q_gram}.}
\begin{proof}
  We will make use of the {\em suffix array} and {\em lcp array}.

  The {\em suffix array}~\cite{manber93:_suffix} $\SA$ of any string $T$
  is an array of length $|T|$ such that
  $\SA[i] = j$, where $T[j:|T|]$ is the $i$-th lexicographically smallest suffix of $T$.
  The \emph{lcp} array of any string $T$ is an array of length $|T|$ such that
  $\LCP[i]$ is the length of the longest common prefix of
  $T[\SA[i-1]:|T|]$ and $T[\SA[i]:|T|]$ for $2 \leq i \leq |T|$, 
  and $\LCP[1] = 0$.

  It is well known that the suffix array for any string of length $|T|$
  can be constructed in $O(|T|)$ time (e.g.~\cite{Karkkainen_Sanders_icalp03})
  assuming an integer alphabet.
  Given the text and suffix array, the lcp array can also be calculated
  in $O(|T|)$ time~\cite{Kasai01}.

  We can calculate the overlapping $q$-gram frequencies of string $T$
  using suffix array SA and lcp array LCP.
  $SA[i]$ represents an occurrence of a $q$-gram $T[SA[i]:SA[i]+q-1]$.
  Since the suffixes are lexicographically sorted in the suffix array,
  intervals on the suffix array where the values of lcp array are at
  least $q$ represent occurrence of the same $q$-gram.
  The sum of $w[SA[i]]$ in this interval is the desired value for the
  $q$-gram.
  Constructing SA, LCP can be done in $O(|T|)$ time,
  and summing up $w[SA[i]]$ for each interval 
  where $LCP[i] \geq q$ can easily be done in $O(|T|)$
  by a simple scan.\qed
\end{proof}

\subsubsection{Proof of Lemma~\ref{lem:left_right_max_nonoverlap}.}
\begin{proof}
  We prove $\NOcc(T[1:i],P) = |\LnOcc(T[1:i],P)|$ by induction on $i$.
  For $i \leq 1$, the statement clearly holds.
  Now, assume that the statement holds for $i < k$, where $k \geq 2$.
  For $i = k$, 
  notice that $0 \leq \NOcc(T[1:k],P) - |\LnOcc(T[1:k],P)\leq 1$, 
  since there can be at most one new occurrence of $P$ ending at
  position $i$, which may or may not be counted for $\NOcc(T[1:k],P)$.
  If we assume on the contrary that the statement does not hold for $i=k$,
  then $\NOcc(T[1:k],P) - \NOcc(T[1:k-1],P) = \NOcc(T[1:k],P) - |\LnOcc(T[1:k],P)| = 1$.
  Since the change was caused by the new occurrence,
  we have $\NOcc(T[1:k]) = \NOcc(T[1:k-|P|]) + 1$.
  By the inductive hypothesis, we have
  $\NOcc(T[1:k-|P|],P) = |\LnOcc(T[1:k-|P|],P)|$.
  Also, $|\LnOcc(T[1:k],P)| = |\LnOcc(T[1:k-|P|],P)| + 1$,
  since the new occurrence does not overlap with any occurrences
  in $\LnOcc(T[1:k-|P|])$.
  This leads to $\NOcc(T[1:k]) = |\LnOcc(T[1:k],P)|$, a contradiction.
  $\NOcc(T, P) = |\RnOcc(T, P)|$ can be shown symmetrically.
  \qed
\end{proof}

\subsubsection{Proof of Lemma~\ref{lem:max_lnocc_lpc}.}
\begin{proof}
We compute the smallest occurrence $b_i$ in $(j-1) \oplus \LnOcc(X_i[j:\be_i], \bstr_{j})$
that crosses $X_\ell$ and $X_r$.
Also, we compute the smallest occurrence $\bb_r$ in $(j-1) \oplus \LnOcc(X_i[j:\be_i], \bstr_{j})$
that is completely within $X_r$.

Then the desired value $\max_1 \LnOcc(X_i[j:\be_i], \bstr_{j})$ can be computed
depending whether $b_i$ and $\bb_r$ exist or not.

Formally, 
consider the set 
$S = ((j-1) \oplus \LnOcc(X_i[j:\be_i], \bstr_j)) \cap [|X_\ell|-q+2:|X_\ell|]$
of occurrence of $\bstr_j$ which is either empty or singleton.
If $S$ is singleton, then let $b_i$ be its single element.
Let
$\bb_r = \min \{k \mid k \in ((j-1) \oplus \LnOcc(X_i[j:\be_i], \bstr_j)) \cap [|X_\ell|+1:|X_\ell|+2(q-1)],
\mbox{if } \exists b_i \mbox{ then } k \geq b_i + q\}$.

Then we have   
\begin{eqnarray*}
&&\textstyle{\max_1} \LnOcc(X_i[j:\be_i], \bstr_{j}) \\
&&=  \begin{cases}
    \max_1 \LnOcc(X_\ell[j:\be_\ell],\bstr_j)
    & \mbox{if } \not \exists b_i \mbox{ and } \not \exists \bb_r \\
    
    b_i-j+1 & \mbox{if } \exists b_i \mbox{ and } \not \exists \bb_r \\
    \bb_r-j + \max_1 \LnOcc(X_r[\bb_r-|X_\ell|:\be_r], \bstr_j) & \mbox{if } \exists \bb_r
  \end{cases}
\end{eqnarray*}
(See also Fig.~\ref{lem:max_lnocc_lpc} in Appendix~\ref{app_sec:figure}.)

For all variables $X_i$ we pre-compute $\prefix(X_i, 3(q-1))$ and $\suffix(X_i, 3(q-1))$.
This can be done in a total of $O(qn)$ time.
If $b_i$ or $\bb_r$ exists,
$|X_\ell|-3(q-1) \leq j-1+ \max \LnOcc(X_\ell[j:\be_\ell], j) \leq |X_\ell|-q+1$.
Then, each $b_i$ and $\bb_r$ can be computed from
$\LnOcc(X_i[(j-1+ \max \LnOcc(X_\ell[j:\be_\ell], j)):|X_\ell|+3(q-1)], p_j)$
runnning the KMP algorithm on string $\prefix(X_i, 3(q-1))\suffix(X_i, 3(q-1))$.

Based on the above recursion, we can compute 
$\max_1 \LnOcc(X_i[j:\be_i], \bstr_{j})$ in a total of $O(q^2n)$
time for all $1 \leq i \leq n$ and $1 \leq j \leq 2(q-1)$.

  It is not difficult to see that similar claims,
  with slightly different conditions, can be made for
  $\max_2\LnOcc(X_i[j:\be_i],\bstr_j)$ where the value corresponds to one of 4 values:
  $\max_2\LnOcc(X_\ell[j:\be_\ell],\bstr_j)$,
  $\max_1\LnOcc(X_\ell[j:\be_\ell],\bstr_j)$,
  $b_i$,
  or $\max_2\LnOcc(X_r[\bb_r-|X_\ell|:\be_r],\bstr_j)$,
  with appropriate offsets.  \qed
\end{proof}

\subsubsection{Proof of Lemma~\ref{lem:max_lnocc_lsc}.}

\begin{proof}
Our basic strategy for computing $\max \LnOcc(X_i[\eb_i:\ee_i], \estr_{j})$
is as follows.
Firstly we compute the largest element of $\LnOcc(X_i[\eb_i:\ee_i], \estr_{j})$
that occurs completely within $X_\ell$.
Secondly we compute the smallest element of $\LnOcc(X_i[\eb_i:\ee_i], \estr_{j})$
that crosses the boundary of $X_\ell$ and $X_r$.
Let $d$ be this occurrence, if such exists.
Then the desired output $\max \LnOcc(X_i[\eb_i:\ee_i], \estr_{j})$ is 
given as either 
the largest or the second largest element of $(d+q-1) \oplus \LnOcc(X_r[d+q-|X_\ell|:|X_r|], \estr_j)$.

More formally:
We consider the case where $\eb_i+q-1 \leq |X_\ell|$.
Let $\ee_\ell = q-1+\max (\Occ(X_i, \estr_j) \cap [|X_\ell|-2(q-1)+1:|X_\ell|-q+1])$,
$m = \eb_i-1+\max \LnOcc(X_\ell[\eb_i:\ee_\ell],\estr_j)$ where 
$(\eb_i, \ee_\ell) = \lsc_q(X_\ell,|X_\ell|-\ee_\ell+1)$.
Let $d =m+q-1+ \min \LnOcc(X_i[m+q:\ee_i],\estr_j)$.
Let
\[
 \bb_r = 
  \begin{cases}
   d & \mbox{if } \ee_i\!-\!q\!+\!1 \! \leq \! |X_\ell| \mbox{ or } d \! > \!|X_\ell|, \\
   d\!+\!q\!-\!1\!+\!\min\LnOcc(X_i[d\!+\!q:|X_i|], \estr_j) & \mbox{otherwise.}
  \end{cases}
\]
Let 
$h^\prime = |X_\ell|+\max_2 \LnOcc(X_r[\bb_{r^\prime}:\be_{r^\prime}], \estr_j)$ and
$h = |X_\ell|+\max_1 \LnOcc(X_r[\bb_{r^\prime}:\be_{r^\prime}], \estr_j)$
where $(\bb_{r^\prime}, \be_{r^\prime}) =\lpc_q(X_r,\bb_r-|X_\ell|)$.
(See also Fig.~\ref{fig:maxLnOcc_eloc} in Appendix~\ref{app_sec:figure}.)
Then 
\[
\max \LnOcc(X_i[\eb_i:\ee_i], \estr_j) = 
 \begin{cases}
  h & \mbox{if } h \leq \ee_i-q+1, \\
  h^\prime & \mbox{otherwise.} \\
 \end{cases}
\]
The case where $\eb_i+q-1 > |X_\ell|$ can be solved similarly.

Each $\ee_\ell$, $d$ and $\bb_r$ can be computed in $O(q)$ time using the KMP algorithm,
hence requiring a total of $O(q^2n)$ time.
By Lemmas~\ref{lem:lpc} and~\ref{lem:lsc},
$\lsc_q(X_\ell,\ee_\ell)$ and $\lpc_q(X_i,\bb_r)$ can be computed in
$O(q^2n)$ time for all $X_i = X_\ell X_r$ and $1 \leq j < n$.
By Lemma~\ref{lem:max_lnocc_lpc},
$h^\prime$ and $h$ can be computed in a total of 
$O(q^2n)$ time for all $X_i = X_\ell X_r$ and $1 \leq j < n$.
Therefore, by dynamic programming we can compute
$\LnOcc(X_i[\eb_i:\ee_i], \estr_j)$ in a total of $O(q^2n)$ time.
\qed
\end{proof}

\clearpage
\section{Figures} \label{app_sec:figure}
\begin{figure}[h]
  \centerline{\includegraphics[width=0.6\textwidth]{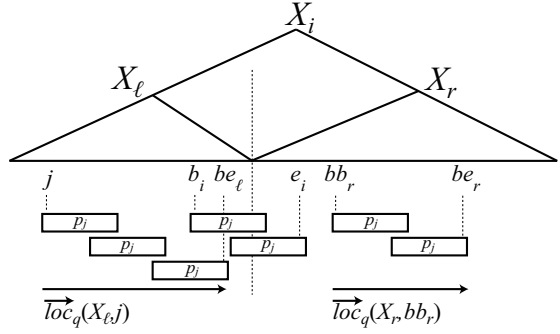}}
  \caption{Illustration for Lemma~\ref{lem:lpc}.
  In this figure, $\lpc_q(X_i, j) = (j, e_i)$.}
  \label{fig:ploc}
\end{figure}

\begin{figure}[h]
  \centerline{\includegraphics[width=0.6\textwidth]{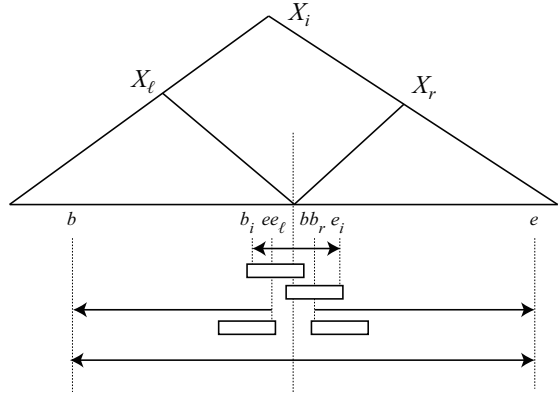}}
  \caption{Illustration for Lemma~\ref{lem:loc_kusi}. Rectangles show important occurrences of $X_i[j:j+q-1]$. In this case $b = \eb_\ell$ and $e = \be_r$.}
  \label{fig:loc_kusi}
\end{figure}

\begin{figure}[h]
  \centerline{\includegraphics[width=0.6\textwidth]{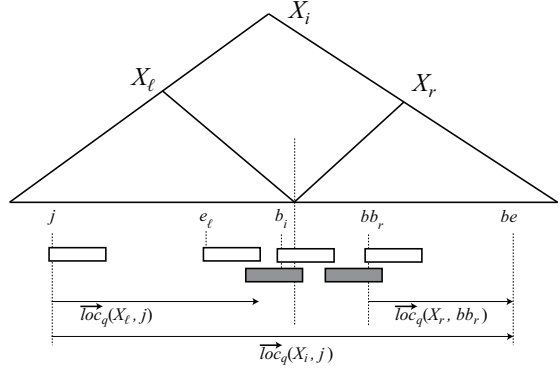}}
  \caption{Illustration for Lemma~\ref{lem:max_lnocc_lpc},
    calculating $\max \LnOcc(X_i[j:\be], \bstr_{j})$.
    Shadowed occurrences are not in $\LnOcc(X_i[j:\be_i], \bstr_{j})$,
    while white ones are in $\LnOcc(X_i[j:\be_i], \bstr_{j})$.}
  \label{fig:maxLnOcc}
\end{figure}

\begin{figure}[h]
  \centerline{\includegraphics[width=0.6\textwidth]{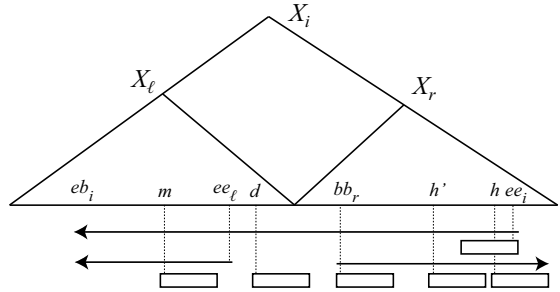}}
  \caption{Illustration for Lemma~\ref{lem:max_lnocc_lsc}. Rectangles show important occurrences of $\estr_j$. In this case $\max\LnOcc(X_i[\eb_i, \ee_i], \estr_j) = h^\prime$, as $h > \ee_i-q+1$.}
  \label{fig:maxLnOcc_eloc}
\end{figure}

\begin{figure}[h]
  \centerline{\includegraphics[width=0.6\textwidth]{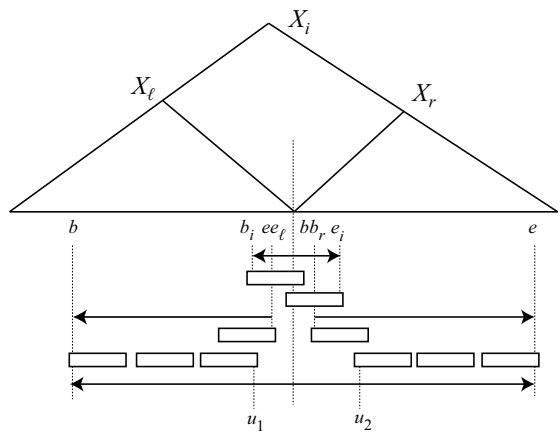}}
  \caption{Illustration for Lemma~\ref{lem:nocc_in_loc}. Rectangles show important occurrences of $X_i[j:j+q-1]$. In this case 
  $\NOcc(X_i[b:\ee_\ell], s_j) = 3$,
  $\NOcc(X_i[u_1+1:u_2-1], s_j) = 1$, and 
  $\NOcc(X_i[\bb_r:e], s_j) = 3$.}
  \label{fig:nocc_loc}
\end{figure}

\end{document}